\documentclass[11pt,b5paper,twoside]{article}      

\newcommand{\mechyear}{2016}
\newcommand{\mechvol}{30}

\usepackage[cp1251]{inputenc}
\usepackage[T2A]{fontenc}
\usepackage[russian]{babel}

\usepackage[dvips]{graphicx}
\usepackage[dvips]{color}

\usepackage{amsmath}
\usepackage{amssymb}
\usepackage{amsxtra}
\usepackage{latexsym}
\usepackage{ifthen}
\usepackage[centerlast,footnotesize,sl]{caption2}
\usepackage{bm}

\usepackage{trudyn}

\newcommand{\vol}{\mathbf {d}}

\begin{document}

\setcounter{page}{96}


\udk{531.38}

\title{ЭНТРОПИЯ ТЕРМОДИНАМИЧЕСКОГО ГРАФА}
      {Энтропия термодинамического графа}


\titleeng{Тhe entropy of a thermodynamic graph}


\author{А.\,Н. Курганский, А.\,Ю. Максимова}
\authoreng{O. Kurganskyy, A.\,J. Maksimova}

\date{27.06.16}




\address{ГУ <<Ин-т прикл. математики и механики>>, Донецк}


\email{kurgansk@gmx.de}


\maketitle

\begin{abstract}
%
%
В работе вводится алгоритмическая модель теплообмена на графе~--
термодинамический граф. Термодинамический граф является аналогом
сеток в методе конечных разностей: вычисление температур
осуществляется в вершинах графа, а ребра графа указывают на
непосредственный теплообмен между вершинами. Рекуррентные
соотношения теплообмена в графе выводятся без обращения к
дифференциальным уравнениям, а только опираясь на коэффициенты
теплопроводности и удельной теплоемкости. Такой подход
представляется авторам одновременно более коротким и гибким с точки
зрения алгоритмического моделирования термодинамического процесса,
чем вывод разностных схем из дифференциальных уравнений. Далее
вводится понятие энтропии термодинамического графа. Доказывается не
улучшаемая в общем случае оценка длины шага по времени, при котором
энтропия  не убывает. Как следствие, данная оценка является
одновременно точной границей устойчивости модели в общем случае.
\vspace{1mm}\\
\textbf{\emph {Ключевые слова: }}\itshape{теплоперенос теплопроводностью, теория графов, диаграмма Вороного, энтропия.}
\end{abstract}

\abstracteng{We introduce an algorithmic model of heat conduction,
the thermodynamic graph.   The thermodynamic graph is analogous to
meshes in the finite difference method in the sense that the
calculation of temperature is carried out at the vertices of the
graph, and the edges indicate the direct heat exchange between the
vertices. Recurrence relations of heat conduction in graph are
derived without using of differential equations and based on the
coefficients of thermal conductivity and heat capacity. This
approach seems to be more direct and flexible from the point of view
of algorithmic modeling of thermodynamic process than the derivation
of difference schemes from differential equations. We introduce also
the notion of entropy of thermodynamic graph. We find the maximum
length of the time step at which the entropy does not decrease in
the general case. As a result, this give us the accurate boundary of
the model stability.}


\kweng{heat transfer by conduction, graph theory, Voronoi diagram,
entropy.\vskip-5mm}


\Section{Введение} Математические модели физических или других
процессов в форме дифференциальных уравнений, основанных на языке
актуальной непрерывности и бесконечности, являются не всегда в
полной мере и ясно интерпретируемыми с алгоритмической точки зрения
при компьютерном моделировании, в котором речь может идти не более
чем о потенциальной, конструктивной непрерывности или бесконечности.
В связи с возникающими таким образом сложностями при построении
алгоритмических моделей физических процессов следует понимать, что
математические формулы с равенствами являются частным случаем
описания алгоритмов. Формулы для дискриминанта квадратного
трехчлена, разложения в ряды, дифференциальные уравнения,
описывающие теплообмен, ассоциативные, коммутативные законы и т.д. и
т.п. являются примерами выражений для вычислений или их
эквивалентных преобразований. Другими словами, конечные
алгебраические формулы, бесконечные ряды, частично-рекурсивные
функции, машины Тьюринга являются, в конце концов, языками для
выражения вычислений при компьютерном моделировании, являющемся во
многом конечной целью исследований. Определенным недостатком языков
частично-рекурсивных функций и машин Тьюринга, как понятий
алгоритма, является трудность их математического анализа
инструментами дифференцирования и интегрирования. Но, с точки зрения
компьютерных вычислений, они имеют явные преимущества в силу своей
природы. Таким образом, нет однозначных преимуществ в компьютерном
моделировании теплообмена, отталкивающегося именно от
дифференциальных уравнений. Это очевидно потому, что сами
дифференциальные уравнения обосновываются предельным переходом в
соотношениях, описывающих процессы теплообмена на элементарных
объемах, а разностные схемы – есть результат обратного процесса
дискретизации дифференциальных уравнений, т.е. возвращения к
элементарным объемам. Таким образом, путь к разностным схемам делает
два лишних шага: 1) предельного перехода от дискретного описания к
непрерывному и 2) обратной дискретизации от дифференциальных
уравнений к разностным схемам. При компьютерном моделировании, если
не ставится задача математического анализа качественных свойств
процесса, достаточно остановиться по пути к алгоритмическому
описанию конкретного процесса теплообмена на шаге до предельного
перехода. Конечно, при ясном понимании того, что влечет за собой
предельный переход, от каких моментов описываемого процесса в
результате этого перехода абстрагируются, а какие новые моменты
возникают уже в непрерывной форме – эти два шага дают новые знания
благодаря математическому анализу. Однако, по опыту компьютерного
моделирования многообразных процессов теплообмена, путь к
вычислительной программе, отталкивающийся не от дифференциальных
уравнений, а от того же, что и вывод самих уравнений, может
оказаться короче.

Исходя из сказанного выше, в данной работе дается определение модели
термодинамического графа без обращения к дифференциальным
уравнениям, исследуются его свойства, связанные с понятием энтропии
и сходимостью численного моделирования термодинамического графа.
Метод конечных разностей с использованием явной схемы является, как
показано в работе, частным случаем такой модели.  Термодинамический
граф служит основой для экспериментов при изучении процессов
теплообмена, рассматриваемых в~[1].

\Section{Термодинамический  граф} Пусть  $G=(V,E)$~-- конечный,
ориентированный, связный, простой граф, где  $V$~-- множество
вершин, $E\subseteq V\times V$~-- множество дуг или, другими
словами, отношение смежности на вершинах. Отношение  $E$ полагаем
симметричным. Напомним, что в случае ориентированного простого графа
принято говорить о дугах, т.е. об упорядоченных парах вершин
$(v,w)$, а в случае неориентированного графа – о ребрах графа, т.е.
о двухэлементных $\{v,w\}$ множествах вершин. Поскольку  $E$
является симметричным, то  $G$ можно рассматривать как
неориентированный граф, но для удобства описания некоторых моментов
теплообмена  $G$ определен как ориентированный граф. Иногда дуги
графа мы будем называть ребрами, когда ориентация дуги не важна.

На вершинах и дугах графа определим следующие функции:
\begin{enumerate}
    \item $u:V\times R^+\rightarrow R$~-- температура вершин (температурное поле),  $[K]$, т.е. $u(v,t)$ есть температура вершины  $v$ в момент времени  $t$;
    \item  $\rho:V\rightarrow R$~-- плотность вещества в вершинах графа,  $\left[{\mathit{\text{кг}}}/{\mathit{\text{м}}^3}\right]$;
    \item  $\vol{}:V\rightarrow R$~-- пространственный объем вершин,  $\left[\mathit{\text{м}}^3\right]$;
    \item $m:V\rightarrow R^{+}$~-- масса вершин графа,  $m(v)=\rho(v)\cdot\vol{}(v)$, $[\mathit{\text{кг}}]$;
    \item $c:V\rightarrow R^{+}\cup \{\infty\}$~-- удельная теплоемкость,
    $\left[{\mathit{\text{Дж}}}/({\mathit{\text{кг}}\cdot\mathit{\text{К}}})\right]$;
    \item  $\mu:V\rightarrow R$~-- теплота кристаллизации, $\left[{\mathit{\text{Дж}}}/{\mathit{\text{кг}}}\right]$;
    \item  $\mathit{melting}:V\rightarrow R$~-- температура плавления,  $[K]$;
    \item  $f:V\times R\rightarrow \{0,1\}$~-- фаза вершин во времени, договоримся, что  $f\left(v,t\right)=1$ означает жидкую фазу вершины  $v$ в момент времени  $t$,  $f\left(v,t\right)=0$~-- твердую;
    \item  $S:E\rightarrow R^+$~-- площади поверхности соприкосновения элементов вещества,
    $\left[\mathit{\text{м}}^2\right]$, т.е., если  $e=(v,w)$, то  $S(e)$~-- площадь поверхности теплообмена между  вершинами $v$ и  $w$;
    \item  $k:E\rightarrow R$~-- коэффициент теплопроводности,
    $\left[{\mathit{\text{Дж}}}/({\mathit{\text{м}}\cdot\mathit{\text{К}}\cdot\mathit{\text{с}}})\right]$;
    \item $\Delta x:E\rightarrow R$~-- расстояние между вершинами,  $[\mathit{\text{м}}]$, т.е., если
$e=(v,w)$, то  $\Delta x(e)$ является расстоянием между  $v$ и  $w$.
\end{enumerate}

Значение температур вершин в момент времени  $t=0$ назовем начальным
условием. Значение температур  вершин  $v$, в которых $c(v)=\infty$,
назовем граничным условием.

Пусть  $v\in V$. Обозначим множество всех смежных с  $v$ вершин в графе через $E(v)=\{w|(v,w) \in E \}$.

\begin{definition}[1]
Граф  $G$ с зафиксированными начальными и граничными условиями, а
также величинами  $c$,  $\rho$, $\vol{}$,  $\mu$,
$\mathit{melting}$,  $f$,  $S$,  $k$,  $\Delta x$, назовем графом
теплообмена или термодинамическим графом. Температурным полем на
графе назовем функцию $u$.
\end{definition}

Геометрические параметры графа теплообмена имеют следующий смысл.
При вычислительном моделировании теплообмена в некотором объеме
вещества, математически представляемого некоторым ограниченным
метрическим пространством $I$, считаем, что вершины графа
соответствуют точкам пространства, в которых вычисляются
температуры. Вершина графа и соответствующая ей точка пространства
$I$ отождествляются. Величина $\Delta x(v,w)$ есть расстояние между
вершинами (узлами)  $v$ и  $w$. Естественный способ привязки к
вершинам геометрических величин  $\vol{}$ и  $S$ следует из
построения диаграммы Вороного~[2] множества вершин графа. А именно,
считаем, что пространство  $I$ разбито на области близости
(многоугольники Вороного), каждая из которых содержит ровно одну
вершину графа и все точки пространства, расстояние от которых до
данной вершины меньше, чем до любой другой вершины графа. Точки
пространства  $I$, лежащие на равном расстоянии от двух или более
вершин графа, образуют поверхность теплообмена этих вершин. Величина
$\vol{}(v)$ есть мера (объем, площадь) многоугольника Вороного,
построенного вокруг вершины  $v$, а значение  $S(v,w)$ есть мера
поверхности соприкосновения многоугольников Вороного, построенных
вокруг  $v$ и  $w$. Иллюстрация к сказанному представлена на
рис.~\ref{fig:1}.

\begin{figure}[ht]
\centering
\includegraphics[scale=0.4, bb= 0mm 0mm 310mm 115mm, clip]{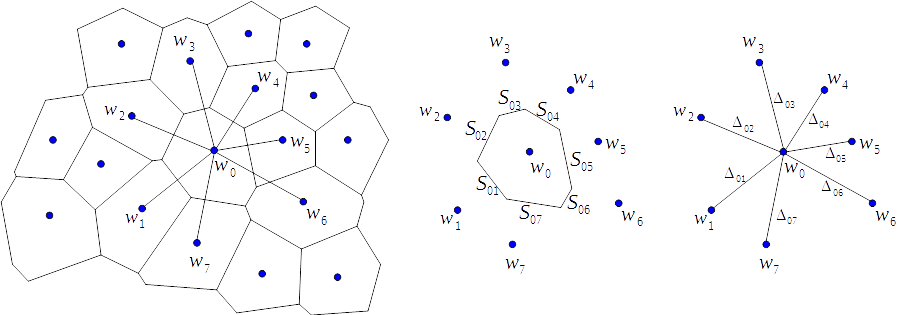}
\caption{Фрагмент диаграммы Вороного с указанием геометрических параметров вершин графа теплообмена}
\label{fig:1}
\end{figure}

 \Section{Уравнение теплообмена} Рассмотрим $v$ графа и
смежную с ней вершину  $w$ в момент времени  $t$. Количество теплоты
получаемой $v$ за время  $\Delta t$ в результате теплообмена с  $w$,
исходя из определения функции $k$, равно

\begin{equation*}
\Delta Q_{(v,w)}=\frac{u(w,t)-u(v,t)}{\Delta x(v,w)}k(v,w)S(v,w)\Delta t.
\end{equation*}
Количество теплоты получаемого вершиной  $v$ за это время в
результате теплового взаимодействия со всеми смежными вершинами
равно
\begin{equation*}
\Delta Q_v=\sum_{w\in E(v)}\frac{u(w,t)-u(v,t)}{\Delta x(v,w)}k(v,w)S(v,w)\Delta t.
\end{equation*}
Учитывая теплоту кристаллизации в случае смены фазы вершины,
получаем
\begin{equation*}
\Delta Q_v=\sum_{w\in E(v)}\frac{u(w,t)-u(v,t)}{\Delta x(v,w)}k(v,w)S(v,w)\Delta{t}+(f(v,t)-f(v,t+\Delta{t}))m(v)\mu(v).
\end{equation*}
С другой стороны, используя коэффициент теплоемкости
$c\left(v\right)$ как связь между изменением энергии вершины и
изменением ее температуры, получаем
\begin{equation*}
 \Delta Q_v=c(v)m(v)(u(v,t+\Delta t)-u(v,t)).
\end{equation*}
Отсюда следует схема для вычисления температуры вершины в следующий момент времени:
\begin{multline}
u(v,t+\Delta t)=u(v,t)+\frac{\Delta t}{c(v)m(v)}\sum_{w{\in}E(v)}\frac{u(w,t)-u(v,t)}{\Delta x(v,w)}k(v,w)S(v,w)+\\
+(f(v,t)-f(v,t+\Delta t))\frac{\mu (v)}{c(v)}.
\end{multline}


При отсутствии смен фаз уравнение теплообмена примет вид
\begin{equation*}
u(v,t+\Delta t)=u(v,t)+\frac{\Delta
t}{c(v)m(v)}\sum_{w{\in}E(v)}\frac{u(w,t)-u(v,t)}{\Delta
x(v,w)}k(v,w)S(v,w).
\end{equation*}

Если  $c(v)={\infty}$,  то $u(v,t+\Delta t)=u(v,t)$, т.е. $v$
является частью граничных условий, при которых температура в  $v$
постоянна.  Множество всех граничных вершин из $V$ обозначим через
${\partial}V$.

Обозначим  $Q(v,t)=c(v)m(v)u(v,t)+f(v,t)m(v)\mu (v)$. Назовем  $Q(v,t)$ тепловой энергией вершины
  $v$ в момент времени  $t$, а величину
\[
Q(G,t)=\sum_{v{\in}V{\partial}V}Q(v,t)
\]
-- тепловой энергией графа  $G$ в момент времени  $t$. Используя
выражение для тепловой энергии в уравнении теплопроводности,
получаем
\begin{equation*}
Q(v,t+\Delta t)=Q(v,t)+\Delta
t\sum_{w{\in}E(v)}\frac{u(w,t)-u(v,t)}{\Delta x(v,w)}k(v,w)S(v,w).
\end{equation*}
\begin{theorem}[1]
Если в графе нет граничных вершин,  ${\partial}V={\emptyset}$, то тепловая энергия  $Q(G,t)$
постоянна во времени.
\end{theorem}
\begin{proof}
Доказательство следует из простого факта
\[
\sum_{v{\in}V}\sum_{w{\in}E(v)}\frac{u(w,t)-u(v,t)}{\Delta x(v,w)}k(v,w)S(v,w)=0,
\]
верного в силу равенств
$S(v,w)=S(w,v)$,  $k(v,w)=k(w,v)$,  $\Delta x(v,w)=\Delta x(w,v)$.
\end{proof}

Далее рассматриваем только уравнение теплообмена без смены фаз.
Пронумеруем вершины  $V$ графа  $G$ числами  $\{1,2,{\dots},|V|\}$.
Пусть  $|V|=n$. Обозначим через  $i$ вершину с номером  $i$, а через
$u_i(t)$ -- температуру вершины  $i$ в момент времени  $t$. Введем
обозначение вектор-столбца:
\[
U(t)=\left(\begin{matrix}u_1(t)\\u_2(t)\\{\dots}\\u_n(t)\end{matrix}\right).
\]

Уравнения теплообмена без фазовых переходов являются линейным
преобразованием поля  $u$. Таким образом, мы можем записать
$U(t+\Delta t)=AU(t)$, где матрица
\begin{equation*}
A=\left(\begin{matrix}\alpha_{11}&{\dots}&\alpha_{1n}\\{\dots}&{\dots}&{\dots}\\\alpha_{n1}&{\dots}&
\alpha_{nn}\end{matrix}\right),
\end{equation*}
которую назовем матрицей теплообмена, обладает очевидными
свойствами, сформулированными в следующей лемме. \pagebreak

\begin{lemma}[1]
Матрица  $A$ обладает следующими свойствами:
\begin{enumerate}
    \item  $\alpha_{ij}=\dfrac{k(i,j)S(i,j)\Delta t}{c(i)m(i)\Delta
x(i,j)}{\geq}0$, если  $j{\in}E(i)$;
    \item  $\alpha_{ij}=0$, если  $j{\notin}E(i)$ и  $j{\neq}i$;
    \item  $\alpha_{ii}=1-\sum_{j{\in}E(i)}\alpha_{ij}$;
    \item $\sum_{j=1}^n\alpha_{ij}=1$;
    \item если  $c(i)<{\infty}$,  $c(j)<{\infty}$, то  $p(i)\alpha_{ij}=p(j)\alpha_{ji}$; \label{lemma:5}
    \item если  $i$ является граничной вершиной, то  $\alpha_{ii}=1$.
\end{enumerate}
\end{lemma}
И справедлива следующая теорема.
\begin{theorem}[2]
При достаточно малых $\Delta t\geq 0$ определитель матрицы
 теплообмена отличен от $0$ и, следовательно, она обратима.
\end{theorem}
\begin{proof}
При $\Delta t=0$  матрица $A$ является единичной. В силу
непрерывности определителя как функции элементов матрицы,
определитель матрицы теплообмена отличен от нуля при достаточно
малых $\Delta t\geq 0$.
\end{proof}

\Section{Энтропия термодинамического графа}
%
Величину $\phi(v)=c(v)\rho(v)$ назовем плотностью теплоемкости в
вершине $v$, а $\phi(v)\vol{}(v)$~-- теплоемкостью вершины $v$.
Величину
\[
\Phi(G)=\sum_{v{\in}V\setminus{\partial}V}\phi(v)\vol{}(v)
\]
назовем теплоемкостью графа, а величину
$p(v)=({\phi(v)\vol{}(v)})/{\Phi(G)}$ -- взвешенной теплоемкостью
вершины  $v$.

Средней температурного поля  $u$ и температурой графа  $G$ в момент
времени  $t$ назовем величину
\[
M_u(t)=\sum_{v{\in}V\setminus{\partial}V}u(v,t)p(v)=
\dfrac{\sum_{v{\in}V\setminus{\partial}V}u(v,t)\phi(v)\vol{}(v)}{\Phi(G)}.
\]
Обозначим через
\begin{equation}\label{eqkurg:D}
D_u(t)=\sum_{v{\in}V\setminus{\partial}V}|u(v,t)-M_u(t)|p(v)=
\dfrac{\sum_{v{\in}V\setminus{\partial}V}|u(v,t)-M_u(t)|\phi(v)\vol{}(v)}{\Phi(G)}
\end{equation}
средневзвешенное абсолютное отклонение температурного поля  $u$. Заметим, что  $D_u(t)=0$ тогда и только
тогда, когда тепловое поле в  $V\setminus{\partial}V$ константное.

Величину
\begin{equation*}
N_u(t)=\frac{D_u(t)}{M_u(t)}\Phi(G)=\frac{\sum_{v{\in}V\setminus{\partial}V}|u(v,t)-M_u(t)|\phi(v)\vol{}(v)}{M_u(t)}
\end{equation*}
назовем негэнтропией поля  $u$ в момент времени  $t$, а величину
\begin{equation*}
\Delta S_u(t)=N_u(t)-N_u(t+\Delta t)
\end{equation*}
-- разностной формой энтропии теплового поля  $u$ в момент времени
$t$. Если тепловое поле является константным, то
$N_u(t){\equiv}\Delta S_u(t){\equiv}0$ для всех  $t$.

\begin{corollary}[1]
Негэнтропия  термодинамического графа равна $0$ тогда и только
тогда, когда тепловое поле на не граничных вершинах графа
константное.
\end{corollary}

\begin{theorem}[3]
Для термодинамического графа выполняются неравенства
\[
0\le N_u(t)<2\Phi(G).
\]
\end{theorem}\vskip-4mm
\begin{proof}
Рассмотрим формулу~(\ref{eqkurg:D}), в которой для упрощения
обозначений опущен параметр $t$ и $M=M_u$:
\begin{equation*}
D_u=\sum_{i{\in}V\setminus{\partial}V}|u(i)-M|p(i).
\end{equation*}
Пусть две вершины $v$  и $w$ такие, что $u(v)>M$, $0<u(w)\le M$.
Рассмотрим тепловое  поле $u'$ на графе $G$, которое совпадает с $u$
во всех вершинах графа, кроме $v$ и $w$. Пусть $u'(v)=u(v)+
({u(w)p(w)})/{p(u)}$ и  $u'(w)=0$. Значения $u'$ подобраны так, что
$M_{u'}=M$. Поскольку
\begin{multline*}
(u(v)-M)p(v)+(M-u(w))p(w)=\\
=\left(u(v)-\frac{u(w)p(w)}{p(v)}-M\right)p(v)+Mp(w)<\\
<\left(u(v)+\frac{u(w)p(w)}{p(v)}-M\right)p(v)+Mp(w),
\end{multline*}\vskip-5mm
то
\begin{multline*}
D_{u}=\\
=(u(v)-M)p(v)+(M-u(w))p(w)+\sum_{i{\in}V\setminus\left({\partial}V\cup\{v,w\}\right)}|u(i)-M_u|p(i)<\\
<\left(u(v)+\frac{u(w)p(w)}{p(v)}-M\right)p(v)+Mp(w)+\sum_{i{\in}V\setminus\left({\partial}V\cup\{v,w\}\right)}|u(i)-M_u|p(i)=\\
=D_{u'}.
\end{multline*}

Таким образом, если температурное поле $u$ графа такое, что есть
вершина $v$,  в которой температура больше $0$ и меньше $M$, то поле
можно преобразовать с сохранением средней температуры так, что
негэнтропия графа увеличится.

Пусть теперь две вершины $v$  и $w$ такие,  что $u(v)>M$, $u(w) >
M$. Рассмотрим тепловое поле $u'$ на графе $G$, которое совпадает с
$u$ во всех вершинах графа, кроме $v$ и $w$. Пусть $u'(v)=u(v)+
({u(w)-M})/({p(u)}) \, p(w)$ и  $u'(w)=M$. Значения $u'$ подобраны
так, что $M_{u'}=M$. Поскольку

\begin{multline*}
\left(u(v)-M\right)p(v)+\left(u(w)-M\right)p(w)=\left(u(v)+\frac{u(w)-M}{p(v)}p(w)-M\right)p(v)+Mp(w),
\end{multline*}
\noindent то 
\begin{multline*}
D_{u}=
\left(u(v)-M\right)p(v)+\left(u(w)-M\right)p(w)+\sum_{i{\in}V\setminus\left({\partial}V\cup\{v,w\}\right)}|u(i)-M_u|p(i)=\\
=\left(u(v)+\frac{u(w)-M}{p(v)}p(w)-M\right)p(v)+0p(w)+\sum_{i{\in}V\setminus\left({\partial}V\cup\{v,w\}\right)}|u(i)-M_u|p(i)=\\
=D_{u'}.
\end{multline*}\vskip-2mm
\noindent Отсюда следует, что максимальная негэнтропия графа
достигается только, если температурное поле такое, что во всех
вершинах, кроме одной, температура равна $0$. Пусть вершина $v$
имеет температуру ${M}/{p(v)}$, а температура во всех других
вершинах равна $0$. Тогда средняя температура поля равна $M$ и
\begin{equation*}
D_{u}=\left(\frac{M}{p}-M\right)p(v)+\sum_{i{\in}V\setminus\left({\partial}V\cup\{v\}\right)}Mp(i)=2M(1-p(v))<2M.
\end{equation*}
Отсюда $N_u(t)=({D_u(t)})/({M_u(t)}) \, \Phi(G)<2\Phi(G)$.
\end{proof}

Исходя из теоремы, естественно дать следующее определение энтропии.
\begin{definition}[2]
Энтропией термодинамического графа $G$ назовем величину
\[S_u(t)=2\Phi(G)-N_u(t).\]
\end{definition}
\vskip-6mm

\begin{corollary}[2]
Справедливы неравенства: $0< S_u(t)\le 2\Phi(G)$.
\end{corollary}

\begin{theorem}[4]
Пусть в термодинамическом графе $G=(V,E)$ средняя температура двух подмножеств его вершин  $V'$ и $V''$,  $V=V'\cup V''$, совпадают. Тогда энтропия графа $G$ равна сумме энтропий его подграфов $G'$ и $G''$, образованных, соответственно, вершинами $V'$ и $V''$.
\end{theorem}
\begin{proof}
Достаточно доказать указанное аддитивное свойство для негэнтропии.
Доказательство прямо следует из ее определения
\[
N_u(t)=\frac{\sum_{v{\in}V\setminus{\partial}V}|u(v,t)-M_u(t)|\phi(v)\vol{}(v)}{M_u(t)}
\]\vskip-4mm
\noindent и равенства
\begin{multline*}
\sum_{v{\in}V\setminus{\partial}V}|u(v,t)-M_u(t)|\phi(v)\vol{}(v)=\\[-2mm]
=\sum_{v{\in}V'\setminus{\partial}V}|u(v,t)-M_u(t)|\phi(v)\vol{}(v)+\sum_{v{\in}V''\setminus{\partial}V}|u(v,t)-M_u(t)|\phi(v)\vol{}(v).
\end{multline*}
\end{proof}

\Section{Сходимость}
Определение сходимости в настоящей работе отличается от данного
в~[3], поскольку определение термодинамического графа и соотношений
теплообмена на нем не отталкивается от дифференциальных уравнений.
Устойчивость модели, как непрерывная зависимость от начальных
условий, следует непосредственно из сходимости модели, определяемой
ниже.

Вычислительную модель графа теплообмена при шаге по времени  $\Delta t$ назовем сходящейся, если для любой вершины  $v$
существует  $\lim_{t\rightarrow {\infty}}u(v,t)$. Условие сходимости в виде ограничений на  $\Delta t$  выводится
из простого соображения, являющегося огрубленной формой второго закона термодинамики: температура в любой
вершине  $v$ в следующий момент времени  $t+\Delta t$ должна быть строго больше минимальной и строго меньше
максимальной температур смежных с ней вершин  $E(v)$ в момент времени  $t$ при отсутствии смены фазы. Это
ограничение выражается неравенством $\frac{\Delta t}{c(v)m(v)}\sum_{w{\in}E(v)}\frac{k(v,w)S(v,w)}{\Delta x(v,w)}<1$, для всех  $v{\in}V$,
или, что то же,  $\sum_{j{\in}E(i)}\alpha_{ij}<1$ для всех  $i$. Или, другими словами,
\begin{equation}
\Delta t<\underset{v{\in}V}{\min}\frac{c(v)m(v)}{\sum_{w{\in}E(v)}{\dfrac{k(v,w)S(v,w)}{\Delta x(v,w)}}}.
\label{eq:2}
\end{equation}

Ниже будет доказана сходимость модели при выполнении условия~(\ref{eq:2}). Заметим, что можно привести пример графа теплообмена
такого, что при
\[
\Delta t=\underset{v{\in}V}{\min}\frac{c(v)m(v)}{\sum_{w{\in}E(v)}\dfrac{k(v,w)S(v,w)}{\Delta x(v,w)}}
\]
модель не сходится. В простейшей такой модели граф состоит из двух вершин, а параметры подобраны так, что
$A=\left(\begin{matrix}0&1\\1&0\end{matrix}\right)$. Таким образом, оценку для  $\Delta t$ в условии~(\ref{eq:2}) для сходимости модели в общем случае нельзя улучшить.
\begin{lemma}[2]
Если верно неравенство~(\ref{eq:2}),  то в матрице  $A$ для всех
$i$ и  $j$  $\alpha_{{ij}}{\geq}0$ и  $\alpha_{{ii}}>0$, а также
значение  $1$ в матрице могут принимать только элементы на главной
диагонали.
\end{lemma}

Доказательство, в силу уже сказанного, очевидно. Заметим, что такая матрица, все элементы которой неотрицательные и
суммы элементов каждой строки равны  $1$, называется стохастической или марковской.

Далее, при упоминании векторов, речь будет идти о нормированном
линейном пространстве с нормой Чебышева, равной максимальному
абсолютному значению компонент вектора.
\begin{theorem}[5]
Если матрица  $A$ теплообмена графа без граничных вершин является стохастической, то негэнтропия не
константного теплового поля $u$, убывая, стремится к  $0$ и, следовательно,  $\Delta S_u(t){\geq}0$.
\end{theorem}

\begin{proof}
Требуется доказать, что  $D_u(t){\geq}D_u(t+\Delta t)$ и
$\lim_{\tau\rightarrow {\infty}}D_u(t+$ $+\tau \Delta t)=0$, где
$\tau {\in}N$. Обозначим век\-тор-стол\-бец теплового поля в момент
времени  $t$ через
 \[
U=\left(\begin{matrix}u_1\\u_2\\{\dots}\\u_n\end{matrix}\right).
\]
Тогда  $M_u(t+\Delta
t)=\sum_{i=1}^np_i\sum_{j=1}^n\alpha_{ij}u_j=\sum_{i=1}^n\sum_{j=1}^np_i\alpha_{{ij}}u_j$.
В силу свойства~\ref{lemma:5} матрицы теплообмена и ее
стохастичности, имеем далее следующие равенства:
\begin{equation*}
\sum_{i=1}^n\sum_{j=1}^np_i\alpha_{{ij}}u_j=\sum_{i=1}^n\sum_{j=1}^np_j\alpha_{{ji}}u_j
=\sum_{j=1}^np_ju_j\sum_{i=1}^n\alpha_{{ji}}=\sum_{j=1}^np_ju_j=M_u(t).
\end{equation*}
Таким образом,  $M_u(t)$ константа при условиях теоремы. Обозначим  $M=M_u(t)$. Тогда
\begin{equation*}
D_u(t+\Delta t)=\sum_{i=1}^np_i|\sum_{j=1}^n\alpha_{{ij}}u_j-M|=\sum
_{i=1}^np_i|\sum_{j=1}^n\alpha_{{ij}}(u_j-M)|{\geq}
\end{equation*}
\begin{equation*}
{\geq}\sum_{i=1}^n\sum_{j=1}^np_i\alpha_{{ij}}|(u_j-M)|=\sum_{i=1}^n\sum_{j=1}^np_j\alpha
_{{ji}}|(u_j-M)|=
\end{equation*}
\begin{equation*}
=\sum_{j=1}^n\sum_{i=1}^np_j\alpha_{{ji}}|(u_j-M)|=\sum_{j=1}^np_j|(u_j-M)|\sum
_{i=1}^n\alpha_{{ji}}=
\end{equation*}
\begin{equation*}
=\sum_{j=1}^np_j|(u_j-M)|=D_u(t).
\end{equation*}

Таким образом, мы доказали, что  $D_u(t+\Delta t){\geq}D_u(t)$.
Покажем, что  $\lim_{\tau \rightarrow {\infty}}D_u(t+$ $+\tau
{\cdot}\Delta t)=0$. Во-первых, очевидно, что $D_u(t)=0$ тогда и
только тогда, когда поле  $u$ константное, т.е. во всех вершинах
значение температуры равно  $M$. Константное поле не меняется со
временем, т.е.  $D_u(t+\tau {\cdot}\Delta t)=0$. Предположим, что в
последовательности  $D_u(t+\tau {\cdot}\Delta t)$ нет нулевых
элементов и $\lim_{\tau \rightarrow{\infty}}D_u(t+\tau {\cdot}\Delta
t)=$ $=d>0$.

Рассмотрим последовательность  $A^{\tau }U$,  $\tau {\in}N$. Заметим, что это ограниченная последовательность и оператор
 $A$ непрерывен, поскольку линейный и конечномерный. Заметим также, что неподвижными точками  $A$ являются только
константные тепловые поля. По известной теореме Боль\-цано--Вей\-ер\-штрасса в  $A^{\tau }U$ можно выбрать сходящуюся
подпоследовательность, пределом которой будет тепловое поле 1) со средневзвешенным абсолютным отклонением равным  $d$ и
2) являющееся неподвижной точкой оператора  $A$. Противоречие.
\end{proof}
\begin{corollary}[3]
Если в графе теплообмена нет граничных вершин и матрица  $A$ стохастическая, то модель сходится к
константному тепловому полю.
\end{corollary}
\begin{corollary}[4]
Если $A$  является матрицей теплообмена графа без граничных вершин, то
\[
\lim_{\tau\rightarrow {\infty}}A^{\tau}=
\left(\begin{matrix}1/n&{\dots}&1/n\\{\dots}&{\dots}&{\dots}\\1/n&{\dots}&1/n\end{matrix}\right).
\]
\end{corollary}

Рассмотрим теперь случай графа с граничными вершинами.
\begin{theorem}[6]
Если матрица  $A$ теплообмена на графе  $G$ с непустым множеством
граничных вершин стохастическая, то модель сходится, причем
\linebreak $\lim_{\tau\rightarrow {\infty}}A^{\tau }U$ зависит
только от значений температуры в граничных вершинах.
\end{theorem}
\begin{proof}
Без ограничений общности положим, что подграф графа  $G$, образованный его неграничными вершинами, является связным. В
противном случае, можно рассмотреть каждую компоненту связности отдельно.

Пусть  $|{\partial}V|=m>0$. Также без ограничений общности для
удобства считаем, что ${\partial}V=\{1,2,{\dots},m\}$. Заметим, что
в этом случае в матрице  $A$ в первых  $m$ строках диагональные
элементы равны  $1$, а все остальные элементы этих строк равны  $0$.
В силу того, что подграф, образованный неграничными вершинами,
связный, справедливо, что через максимум  $n-1$ шагов по времени
каждая неграничная вершина графа испытывает влияние всех других
вершин, т.е. в матрице  $A^{\tau }$,  $\tau {\geq}n-1$, в каждой
строке с номером больше  $m$ все элементы строго больше $0$.

Рассмотрим пространство  $U$-мер\-ных векторов, в которых первые $m$
элементов равны $0$. Это пространство замкнуто относительно
оператора  $A$. Покажем, что  $A^{\tau}=(\alpha_{ij}^{\tau})$, для
любого  $\tau{\geq}n-1$, является сжимающим отображением в  $U$,
т.е. существует  $0{\leq}\lambda_{\tau}<1$, что для любого $V{\in}U$
справедливо  $\|A^{\tau }V\|{\leq}\lambda_{\tau}\|V\|$. Под нормой
понимаем норму Чебышева, т.е. максимальное абсолютное значение
координат вектора. Действительно, пусть
$\lambda_{\tau}=\underset{m<i{\leq}n}{\max}|\sum_{j=m+1}^n\alpha_{{ij}}^{\tau
}|$. Очевидно, что $\lambda_{\tau }<1$, поскольку сумма элементов
любой строки матрицы не больше  $1$ и все элементы рассматриваемых
строк строго больше  $0$, а в сумме участвует только часть элементов
строк. Учитывая, что первые  $m$ координат векторов из  $U$ равны
$0$, получаем для  $V{\in}U$
\[
\|A^{\tau }V\|=\underset{m<i{\leq}n}{\max}|\sum_{j=m+1}^n\alpha_{{ij}}^{\tau
}v_j|{\leq}\|V\|{\cdot}\underset{m<i{\leq}n}{\max}|\sum_{j=m+1}^n\alpha
_{{ij}}^{\tau }|=\lambda_{\tau }\|V\|.
\]

Рассмотрим множество всех векторов вида  $M_i=\{A^iU-A^jU| i<j\}$. Радиусом
$\|M_i\|$ множества  $M_i$ векторов назовем максимальную норму среди его векторов. Заметим, что
$M_{i+1}=AM_i$. В силу того, что  $A^{n-1}$ сжимающее отображение, радиус  $\|M_i\|$ стремится к  $0$ при
$i\rightarrow\infty$. Как известно, это является необходимым и достаточным условием существования  $\lim_{\tau\rightarrow {\infty}}A^{\tau }U$.

Представим теперь  $U$ как сумму двух векторов $U=U'+U''$ таких, что
в  $U'$ первые  $m$ компонент совпадают с соответствующими
компонентами  $U$, а остальные равны  $0$. Другими словами, в
$U'$~-- температура в граничных вершинах, \linebreak $U''$ ~--
температура в остальных. Поскольку  $U''{\in}U$, то
$\lim_{\tau\rightarrow{\infty}}A^{\tau }U''=0$, где  $0$~-- нулевой
вектор. Таким образом, $\underset{\tau\rightarrow{\infty}}{\lim
}A^{\tau}U=\lim_{\tau\rightarrow {\infty}}A^{\tau}U'$.
\end{proof}
\vskip2mm

\Section{Частный случай} Если вершины графа расположить не случайным
образом, а в узлах конечно-разностной сетки, диаграмма Вороного
будет состоять из ячеек-прямоугольников, каждая (внутренняя – не
лежащая на границе рассматриваемой области) из которых будет иметь
четыре соседние ячейки. Соединим вершины соседних ячеек ребрами
(рис.~\ref{figkur:2}).

\begin{figure}[ht]
\centering
\includegraphics[scale=0.4, bb= 0mm 0mm 100mm 100mm, clip]{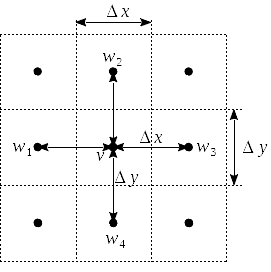}
\caption{Метод конечных разностей как частный случай теплообмена на графе}
\label{figkur:2}
\end{figure}
Обозначениям в уравнении теплообмена на графе соответствуют следующие величины:
\begin{equation*}
\Delta x(v,w_1)=\Delta x(v,w_3)=\Delta x
\end{equation*}
\begin{equation*}
\Delta x(v,w_2)=\Delta x(v,w_4)=\Delta y
\end{equation*}
\begin{equation*}
\vol{}(v)=\Delta x{\cdot}\Delta y
\end{equation*}
\begin{equation*}
S(v,w_1)=S(v,w_3)=\Delta y
\end{equation*}
\begin{equation*}
S(v,w_2)=S(v,w_4)=\Delta x.
\end{equation*}
Переписывая основное уравнение (1) в приведенных обозначениях,
получим явную конеч\-но-разност\-ную схему:
\begin{multline}
u(v,t+\Delta t)=u(v,t)+\Delta t\frac k{c{\cdot}\rho}\left(\frac{u(w_1,t)-2u(v,t)+u(w_3,t)}{\Delta x^2}\right.+\\
+\left.\frac{u(w_2,t)-2u(v,t)+u(w_4,t)}{\Delta y^2}\right).
\end{multline}

\Section{Заключение} В работе определена алгоритмическая модель
теплообмена на графе без опоры на дифференциальные уравнения.
Коне\-чно-раз\-ност\-ная явная схема (4) является частным случаем
такой модели. Для процесса теплообмена без смены фаз доказана не
улучшаемая  в общем случае оценка длины шага по времени, при котором
модель (1) сходится.  При отсутствии граничных вершин доказательство
сходимости основывается на определенном в работе понятия энтропии
термодинамического графа и его свойстве возрастать при указанном
временном шаге. При наличии граничных вершин сходимость доказывается
исходя из свойства матрицы теплообмена быть марковской.

Благодаря тому, что современные языки программирования,  как
например Python, содержат библиотеки для построения диаграмм
Вороного  произвольного множества точек, термодинамический граф
обладает определенной гибкостью в его использовании при
моделировании теплообмена в средах со сложной геометрической
конфигурацией как самой среды, так и граничных условий. Более того,
вместо регулярной, прямоугольной сетки узлов  расчета температур,
можно использовать сетку произвольной геометрии. В частности, в
дальнейших работах предполагается рассмотреть сетки узлов,
расположенных случайным образом с равномерной плотностью
распределения, с целью проведения сравнительных численных
экспериментов в прямоугольных и изотропных сетках узлов.

В дальнейшем предполагается применить термодинамический граф для компьютерного моделирования теплообмена со сменой фаз, для чего  и разрабатывалась модель~[4].

\end{document}